\documentclass{article}
\usepackage[T1]{fontenc}
\usepackage{graphicx}
\usepackage[latin9]{inputenc}
\usepackage{amsmath,mathrsfs,mathtools,amsthm}
\usepackage{amsmath,amsfonts,amssymb,epsfig,multicol,multirow}
\usepackage{rotating,hhline,stmaryrd}
\usepackage{float}
\usepackage{array}
\usepackage{adjustbox}

\newtheorem{theorem}{Theorem}

\newtheorem{lemma}{Lemma}

\newtheorem{defn}{Definition}
\newtheorem{prop}[theorem]{Proposition}
\newtheorem{cor}[theorem]{Corollary}

\newcommand{\ben}{\begin{equation*}}
\newcommand{\een}{\end{equation*}}

\newcommand{\comment}[1]{}

\begin{document}

\title{DNA codes over two noncommutative rings of order four}

\author{Jon-Lark Kim\thanks{Department of Mathematics,
Sogang University,
Seoul 04107, South Korea.
{Email: \tt jlkim@sogang.ac.kr},
Corresponding author,
}

~~Dong Eun Ohk\thanks{
             Department of Mathematics, Sogang University, Seoul 04107, Republic of Korea}
}

\date{}

\maketitle

\begin{abstract}
In this paper, we describe a new type of DNA codes over two noncommutative rings $E$ and $F$ of order four with characteristic 2. Our DNA codes are based on quasi self-dual codes over $E$ and $F$. Using quasi self-duality, we can describe fixed GC-content constraint weight distributions and reverse-complement constraint minimum distributions of those codes.
\end{abstract}

\noindent
{\bf keywords} Coding theory, DNA codes, quasi-self dual codes, rings

	\section{Introduction}

L. Adleman~\cite{Adleman} performed the computation using DNA strands to solve an instance of the Hamiltonian path problem giving birth to DNA computing. Since then, DNA computing and DNA storage have been developed. This development requires several theories for the construction of DNA sequences satisfying various constraints. Algebraic coding theory has contributed to construct DNA codes with constraints. See  Limbachiya et al.~\cite{Limbachiya} for up-to-date DNA codes over various rings

DNA codes based on error-correcting codes have been successful in DNA-based computation and storage. For example, Milenkovic~\cite{Mil} described the design of codes for DNA computing by considering secondary
structure formation in single-stranded DNA molecules and non-selective
cross-hybridization.

Since there are four nucleobases in DNA, two well known algebraic structures such as the finite field $GF(4)$ and the integer modular ring $\mathbb{Z}_4$ have been used because they have been well studied in coding community. Besides those with four elements, there has been research on DNA codes over $\mathbb F_2 + u\mathbb F_2$ where $u^2=0$, $\mathbb F_2 + v \mathbb F_2$ where $v^2=v$, or $\mathbb F_2[u]/(u^2-1)$. For example, the papers~\cite{Guenda},~\cite{Liang}, ~\cite{Siap} use two rings $\mathbb F_2 + u\mathbb F_2$ and $\mathbb F_2[u]/(u^2-1)$. They only consider cyclic DNA codes over these rings. Even though some cyclic DNA codes over these rings are reversible or reversible-complementary under certain conditions, the length $n$ is restricted. Our tables in Page 18 and 19 contain all possible parameters of quasi-self-dual DNA codes with highest minimum weight for a given fixed GC-weight. Certainly this is a new result.

We note that these rings are all commutative rings with unity. From calculation point of view, certainly these rings have an advantage. On the other hand, due to various possibilities of DNA sequences, it is natural to ask whether there are other algebraic structures consisting of four elements.

According to literature~\cite{Fine}, B. Fine classified the 11 finite rings of order four. We observe that the reverse-complement condition of DNA sequences can be translated as a product of multiplication in a ring and prove that exactly 8 rings of order four out of 11 can be studied. These 8 rings include exactly two noncommutative rings with no unity, denoted by $E$ and $F$ in the notation of~\cite{Fine}. This is one of reasons studying DNA codes over $E$ and $F$.

Another reason is motivated by the recent work on quasi-self-dual codes over $E$~\cite{Alahmadi1}. The authors~\cite{Alahmadi1} studied the weight enumerators of these codes by means of invariant theory. Hence the GC weight enumerator of a quasi-self-dual code over $E$ can be derived from the complete weight enumerator of a quasi-self-dual code over $E$ (see Theorem~\ref{thm-GC}). Noting that in general it is difficult to compute the GC weight distribution of a DNA code, we apply a quasi-self-dual code over $E$ to a DNA code to find a subcode of the DNA code which has a fixed GC-content constraint.


This paper consists of six sections. In Section 2, we introduce DNA codes and some definitions of DNA codes. All finite rings of order four will be covered and we define some generalized maps on DNA codes. In Section 3 and 4, we construct Quasi self-dual (QSD) DNA codes based on the QSD codes over $E$ which was considered in~\cite{Alahmadi1}. We also calculate important values of QSD DNA codes over $E$, including the number of inequivalent codes, the GC-weight distribution and the minimum distance of a fixed GC-content subcode on reverse-complement constraints. In Section 5, we  also define QSD DNA codes over $F$ and  compute GC-weight distributions. The classification of QSD DNA codes with $n \le 8$ are listed
in the table at the end. This improves the previous classification of QSD codes with $n \le 6$ by~\cite{Alahmadi1}. Section 6 concludes our paper.

\section{Preliminaries}

\subsection{DNA coding theory}

DNA coding theory is concerned about designing nucleic acid systems using error-correcting codes. DNA, deoxyribo nucleic acid, is a molecule composed of double strands built by paring the four units, Adenine, Thymine, Guanine, and Cytosine, denoted by $A, T, G$ and $C$ respectively, which are called nucleotides. These nucleotides are joined in chains which are bound together with hydrogen bounds. $A$ and $T$ have 2 hydrogen bonds while $G$ and $C$ have 3 hydrogen bonds. Thus these joints make complementary base pairings which are $\{A,T\}$ and $\{G,C\}$. It is called the {\em Watson-Crick complement}. We denote it by $A^C=T$ and $G^C=C$, or equivalently $T^C=A$ and $C^C=G$. So this complement map is a bijection on the set $\{A,T,G,C\}$.

A DNA sequence is a sequence of the nucleotides. The ends of a DNA sequence are chemically polar with $5'$ and $3'$ ends, which implies that the strands are oriented. Given a sequence with the orientation $5' \to 3'$, the reverse complementary is involved naturally. For instance, a DNA sequence
\begin{center}
$5'-TCGGCAACATG-3'$
\end{center}
has its complement
\begin{center}
$3'-AGCCGTTGTAC-5'$.
\end{center}
If we arrange these sequences to have the same orientation, then there are two sequences
\begin{center}
$5'-TCGGCAACATG-3'$ and $5'-CATGTTGCCGA-3'$.
\end{center}
Note that one sequence is the reverse complement of the other.

A DNA sequence means one strand of DNA. The set of DNA strands are needed for DNA computing. Thus we define the DNA code as a fixed set of sequences consisting of $A,T,G,C$, which are also called codewords.

\begin{defn}{\rm
An {\em $(n,M)$ DNA code} $\mathcal{C}$ is a set of codewords of length $n$, size $M$ over four alphabets, $A,T,G,C$. A {\em DNA codeword}, or a {\em DNA sequence} is a codeword of a DNA code.
}
\end{defn}
In general, a DNA code does not need to have algebraic structures. However, in some DNA computation and DNA storage, an error-correction is required. Furthermore to use Algebraic coding theory we expect that the set of DNA sequences has an algebraic structure. Therefore in DNA coding theory, we identify the set $\{A,T,G,C\}$ with an order 4 ring. Since there are 4 types of nucleotides, DNA codes can be constructed from algebraic structures of rings with 4 elements such as the finite field $GF(4)$ or the finite ring $\mathbb{Z}_4$. In fact, any code over 4 elements can be a DNA code, but it will be difficult to analyze properties.

\begin{defn}{\rm
Let $\mathrm{x}=(x_1x_2\cdots x_n)$ be given ($i.e., \; x_i \in \{A,T,G,C\}$).
\begin{enumerate}
\item The {\em reverse} of $\mathrm{x}$, denoted by ${\mathrm{x}}^R$, is the codeword $(x_n x_{n-1}\cdots x_1)$.
\item The {\em complement} of $\mathrm{x}$, denoted by ${\mathrm{x}}^C$, is the codeword $({x_1}^C{x_2}^C\cdots {x_n}^C)$.
\item The {\em reverse complement} of $\mathrm{x}$ is ${\mathrm{x}}^{RC} = {({\mathrm{x}}^R)}^C = {({\mathrm{x}}^C)}^R$.
\end{enumerate}
}
\end{defn}

We can easily check that ${({\mathrm{x}}^R)}^C = {({\mathrm{x}}^C)}^R$ for any DNA sequence $\mathrm{x}$. Using these definitions, we can give constraints on DNA codes.

\begin{defn}{\rm
Let $\mathcal{C}$ be a DNA code, and $d_H$ be the Hamming distance.
\begin{enumerate}
\item The code $\mathcal{C}$ has the {\em reverse constraint} if there exists $d\ge0 $ such that $d_H(x^R,y) \ge d$ for all $x,y \in \mathcal{C}$.
\item The code $\mathcal{C}$ has the {\em reverse-complement constraint} if there exists $d\ge0 $ such that $d_H(x^{RC},y) \ge d$ for all $x,y \in \mathcal{C}$.
\end{enumerate}
}
\end{defn}

Note that the reverse map $\_^R : F^n \to F^n$ is not a linear map ($F$ is a 4-element ring). This reverse map is a permutation, so that it is not independent of the permutation equivalence. For example, let $\mathcal{C} = \{ ATTC, CGGA\}$. Then ${(ATTC)}^R = (CTTA)$, ${(CGGA)}^R = (AGGC)$, so that $d_H(x^R,y) \ge 2$ for all $x,y \in \mathcal{C}$. A permuted code $\mathcal{C}^{\prime} = \{ ATCT, CGAG\}$ has $\min\{d_H(x^R,y)\} = 4$ since $\mathcal{C}^{\prime R} = \{TCTA, GAGC\}$. So for these reverse constraints, we do not consider the permutation equivalence.

In genetics, it is required to compute the $GC$-content. The $GC$-content is the percentage of $G$ and $C$ in a DNA. Since $GC$ pair is held by 3 hydrogen bonds and $AT$ pair is held by 2 hydrogen bonds, high $GC$-content DNAs are more stable than low $GC$-content DNAs. On the other hand, if $GC$-content is too high, then it is difficult to occur DNA replication. Therefore we need to set a proper $GC$-content. In DNA coding theory, we define the $GC$-content as the number of coordinates of $G$ and $C$.
\begin{defn}{\rm
Let $\mathcal{C}$ be a DNA code and $\mathrm{x}$ be a codeword in $\mathcal{C}$.
\begin{enumerate}
\item The {\em $GC$-content} of $\mathrm{x}$ is the number of $G$ and $C$ in $\mathrm{x}$.
\item The DNA code $\mathcal{C}$ has a fixed {\em $GC$-content constraint} if each codeword in $\mathcal{C}$ has the same $GC$-content.
\end{enumerate}
}
\end{defn}

Many codes do not satisfy the fixed $GC$-content constraint. For the fixed $GC$-content constraint, we need to calculate the set of codewords which have the same $GC$-content. Therefore we need the $GC$-weight enumerator.

\begin{defn}{\rm
Let $\mathcal{C}_1$ be a code over 4 elements $\{a_1,a_2,a_3,a_4\}$ and $\mathcal{C}_2$ be a DNA code.
\begin{enumerate}
\item The {\em complete weight enumerator of the code} $\mathcal{C}_1$, $CWE_{\mathcal{C}_1}(w,x,y,z)$ is defined by
\begin{equation*}
CWE_{\mathcal{C}_1}(w,x,y,z) = \sum_{c\in {\mathcal{C}_1}} w^{n_{a_1}(c)}x^{n_{a_2}(c)}y^{n_{a_3}(c)}z^{n_{a_4}(c)}
\end{equation*}
where $n_\alpha(c)$ is the number of occurrences of $\alpha \in \{a_1,a_2,a_3,a_4\}$ in a codeword $c$.
\item The {\em $GC$-weight enumerator of the code} $\mathcal{C}_2$, $GCW_{\mathcal{C}_2}(x,y)$ is the weight enumerator that counts the number of coordinates in $\{G,C\}$ and $\{A,T\}$, which is defined by
\begin{equation*}
GCW_{\mathcal{C}_2}(x,y)=CWE_{\mathcal{C}_2}(x,x,y,y) = \sum_{c\in {\mathcal{C}_2}} x^{n_{G}(c)}x^{n_{C}(c)}y^{n_{A}(c)}y^{n_{T}(c)}
\end{equation*}
\end{enumerate}
}
\end{defn}

We can get the size of a subcode $\mathcal{C}^{\prime}$ which has a fixed $GC$-content using the given polynomial $GCW_{\mathcal{C}}(x,y)$. If $GCW_{\mathcal{C}}(x,y) = \sum a_i x^i y^{n-i}$, then the subcode has the order $| \, \mathcal{C}^{\prime} \, | = a_k$ where $\mathcal{C}^{\prime} = \{ c \in \mathcal{C} \, | \, c \mbox{ has a fixed $GC$-content $k$ (or $n-k$)} \}$.

\subsection{Finite rings of order 4}

Since DNA codes are based on 4 elements, we need 4 elements set. For example, a quaternary code, ${GF(4)}^n$, is a linear code defined over 4 elements field $GF(4)$. This code can be identified a DNA code. Most DNA codes are constructed using $GF(4)$ or $\mathbb{Z}_4$, which are typical rings of 4 elements. However, B. Fine classified rings of order $p^2$ up to isomorphism and so there are 11 finite rings of 4 elements~\cite{Fine}. It is possible to construct DNA codes using other finite rings. The main goal of this paper is to construct DNA codes over other rings except for $GF(4)$ and $\mathbb{Z}_4$, especially the ring $E$.
\vskip 5mm

\begin{table}[h]
\begin{center}
\begin{tabular}{ | c | c | c | c |}
\hline
ring name & ring presentation & char \\
\hline
$A$ & $\left\langle a; \, 4a=0,a^2=a \right\rangle$ & 4 \\
\hline
$B$ & $\left\langle a; \,4a=0,a^2=2a \right\rangle$ & 4 \\
\hline
$C$ & $\left\langle a; \,4a=0,a^2=0 \right\rangle$ & 4 \\
\hline
$D$ & $\left\langle a,b; \,2a=2b=0,a^2=a, b^2=b, ab=ba=0 \right\rangle$ & 2 \\
\hline
$E$ & $\left\langle a,b; \,2a=2b=0,a^2=a, b^2=b, ab=a, ba=b \right\rangle$ & 2 \\
\hline
$F$ & $\left\langle a,b; \,2a=2b=0,a^2=a, b^2=b, ab=b, ba=a \right\rangle$ & 2 \\
\hline
$G$ & $\left\langle a,b; \,2a=2b=0,a^2=0, b^2=b, ab=ba=a \right\rangle$ & 2 \\
\hline
$H$ & $\left\langle a,b; \,2a=2b=0,a^2=0, b^2=b, ab=ba=0 \right\rangle$ & 2 \\
\hline
$I$ & $\left\langle a,b; \,2a=2b=0,a^2=b, ab=0 \right\rangle$ & 2 \\
\hline
$J$ & $\left\langle a,b; \,2a=2b=0,a^2=b^2=0 \right\rangle$ & 2 \\
\hline
$K$ & $\left\langle a,b; \,2a=2b=0,a^2=a, b^2=a+b, ab=ba=b \right\rangle$ & 2 \\
\hline
\end{tabular}
\caption{Classification table of finite rings of order 4}
\end{center}
\end{table}

\begin{table}[h]
\begin{center}
\begin{tabular}{ c | c | c | c | c | c |}
+ & 0 & $a$ & $2a$ & $3a$ \\
\hline
0 & 0 & $a$ & $2a$ & $3a$ \\
\hline
$a$ & $a$ & $2a$ & $3a$ & 0 \\
\hline
$2a$ & $2a$ & $3a$ & 0 & $a$ \\
\hline
$3a$ & $3a$ & 0 & $a$ & $2a$ \\
\hline
\end{tabular} \qquad
\begin{tabular}{ c | c | c | c | c | c |}
+ & 0 & $a$ & $b$ & $c$ \\
\hline
0 & 0 & $a$ & $b$ & $c$ \\
\hline
$a$ & $a$ & 0 & $c$ & $b$ \\
\hline
$b$ & $b$ & $c$ & 0 & $a$ \\
\hline
$c$ & $c$ & $b$ & $a$ & 0 \\
\hline
\end{tabular}
\caption{Addition tables of the ring of order 4; char 4 and char 2}
\end{center}
\end{table}

\begin{prop} The following propositions hold.
\begin{enumerate}
\item $A \cong \mathbb{Z}_4$, $D \cong (\mathbb{Z}_2 + \mathbb{Z}_2) \cong \mathbb{Z}_2[u]/(u^2-u)$, $G \cong \mathbb{Z}_2[u]/(u^2-1) \cong \mathbb{Z}_2[u]/(u^2)$ and $K \cong GF(4)$. These rings have a multiplicative identity. The other rings do not have a multiplicative identity.
\item In the above rings, only $E$ and $F$ are non-commutative rings.
\item Any product of two elements in $C$ or $J$ is zero.
\end{enumerate}
\end{prop}

\begin{proof}
By Table 1, it is clear that (ii) and (iii) hold. So it remains to show that the isomorphisms exist.

Define the following homomorphisms:

$\phi_A : A \to \mathbb{Z}_4$ by $\phi_A(a)=1$

$\phi_D : D \to \mathbb{Z}_2 + \mathbb{Z}_2$ by $\phi_D(a)=(1,0)$ and $\phi_D(b)=(0,1)$

$\phi_G : G \to \mathbb{Z}/(u^2-1)$ by $\phi_G(a)=1+u$ and $\phi_D(b)=1$

$\phi_K : K \to GF(4)= \{0,1,w,1+w\}$ by $\phi_K(a)=1$ and $\phi_D(b)=w$

Then we can easily check the homomorphisms are isomorphisms.
\end{proof}

As mentioned above, most DNA codes are constructed using $A \cong \mathbb{Z}_4$ or $K \cong GF(4)$. J. Liand and L. Wang constructed the cyclic DNA codes, using the ring $\mathbb{F}_2+u\mathbb{F}_2$ with $u^2=0$~\cite{Liang}. This ring $\mathbb{F}_2+u\mathbb{F}_2 \cong G$. K. Guenda and T. Gulliver constructed DNA codes over the same ring $\mathbb{F}_2+u\mathbb{F}_2$ with $u^2=0$~\cite{Guenda}. I. Siap et al. used the ring $\mathbb{F}_2[u]/(u^2-1) \cong G$~\cite{Siap}. N. Bennenni et al introduced another cyclic DNA codes over the ring $\mathbb{F}_2+v\mathbb{F}_2$ with $v^2=v$~\cite{Bennenni}. This ring $\mathbb{F}_2+v\mathbb{F}_2 \cong D$. Even though there are some papers algebraic codes over the rings $E$, $H$ and $I$, the DNA codes over those rings have not been constructed. Thus we focus on the other rings.

Before construct DNA codes over rings, we need to define $maps$ which can calculate easily the complement and the $GC$-content. It means that the complement map and the $GC$-content map should be defined over finite rings.
\begin{defn}{\rm
Let $R$ be a ring of order 4 and $f:\{A,T,C,G\} \to R$ be a proper representation map. It means $f$ is bijective. A {\em complement map} $\phi$ over $R$ is a bijection defined by $\phi(f(x)) = f(x^C)$.
}
\end{defn}

We can check that $\phi(x) \ne x$ and $\phi^2(x) = x$ since $x^C \ne x$ and ${(x^C)}^C = x$. We denote $\phi$ by $\phi(x)=x^C$. This $\phi$ is a bijection on R, so we can define this map $\phi$ easily. The question is whether a simple definition of $\phi$ exist. To be specific we want to define an element $\alpha \in R$ satisfying $x^C = x+\alpha$. By addition table of rings, we can find such $\alpha$ in any finite ring of order 4. If the finite ring has char 4, define $x^C=x+2a$. If the finite ring has char 2, any element of $R$ can be $\alpha$.

Let $\mathcal{C}$ be an additive code over $R$ which is a ring of order 4 and let $\mathrm{x} = (\alpha \alpha \cdots \alpha)$. We can calculate $\mathrm{y}^C$ by $\mathrm{y}^C = \mathrm{y} + \mathrm{x}$. If the codeword $\mathrm{x} \in \mathcal{C}$, then the code $\mathcal{C}^C = \{ \mathrm{c}^C \, | \, \mathrm{c} \in \mathcal{C} \}$ is the same as the original code $\mathcal{C}$. So it implies that the reverse-complement constraint and the reverse constraint are equal in $\mathcal{C}$.

\begin{defn}{\rm
Let $R$ be a ring of order 4. A {\em $GC$-content map} $\psi:R \to GF(2)$ is a function defined by
\begin{equation*}
\psi(x) = \begin{cases}
1, & \mbox{if } f^{-1}(x) = C \mbox{ or } G \\
0, & \mbox{otherwise}
\end{cases}
\end{equation*}
where $f:\{A,T,C,G\} \to R$ is a bijection.
}
\end{defn}

The map $\psi$ can be defined $\psi:R \to A$ where $A = \{0,r\}$ and $r(\ne 0) \in R$. This definition can be extended to $\psi:R \to A \xhookrightarrow{} GF(2)$. For the $GC$-content map we also want to define simply as $\psi(x) = \beta x$ for some $\beta \in R$.

\begin{prop}\label{prop-GC}
We can define the natural $GC$-content map over finite ring of order 4, except the ring $C, J$ and $K$.
\end{prop}

\begin{proof}
Note that this $\beta$ satisfies $\beta x = \beta y$ for some $x \ne y$. If we define $\psi(x) = \beta x$ as follows
\begin{equation*}
\beta = \begin{cases}
2a, & \mbox{ring $A$, and } 0^C = 2a \\
a \mbox{ or } 3a, & \mbox{ring $B$, and } 0^C = 2a \\
a, & \mbox{ring $D$, and } 0^C = b \\
a \mbox{ or } b \mbox{ or } c, & \mbox{ring $E$, and } 0^C = c \\
a, & \mbox{ring $G$, and } 0^C = a \\
b \mbox{ or } c, & \mbox{ring $H$, and } 0^C = a \\
a \mbox{ or } c, & \mbox{ring $I$, and } 0^C = b \\
\end{cases}
\end{equation*}
then the map $\psi$ is well-defined.

In the ring $F$, there is no element $\beta$ satisfying $\beta x = \beta y$ for some $x \ne y$. However, if we define $\psi(x) = x \beta$ where $\beta = a \mbox{ or } b \mbox{ or } c$ and $0^C = c$, then this $\psi$ is well-defined.

Since the rings $C$ and $J$ satisfy $x \cdot y = 0$ for all $x,y$, so the element $\beta$ and the map $\psi$ satisfying $\psi(x) = \beta x$ do not exist.

The ring $K$ is a field, so there is no element $\beta$ satisfying $\beta x = \beta y$ for some $x \ne y$, except zero.

Therefore the ring $C$, $J$ and $K$ do not have the natural $GC$-content map.
\end{proof}

The complement map $\phi$ and the $GC$-content $\psi$ can be defined on a DNA code $\mathcal{C}(n,M)$ and its codeword $\mathrm{x} = (x_1 \cdots x_n)$ by $\overline{\phi}(\mathrm{x}) = (\phi(x_1) \cdots \phi(x_n))$ and $\overline{\psi}(\mathrm{x}) = (\psi(x_1) \cdots \psi(x_n))$. Therefore $\overline{\phi}(\mathcal{C})$ is another $DNA$ code and $\overline{\psi}(\mathcal{C})$ is a binary code. From now on let $\phi$ and $\psi$ be the map on a code. Note that $d_H(\psi(\mathrm{x}))$ is the number of $GC$s in the codeword $\mathrm{x}$. So it is natural that $\psi$ is called the $GC$-content map.

By Proposition~\ref{prop-GC}, many four element rings have a natural $GC$-content map. However, the Galois field $GF(4)$ does not have, even though it is widely used. Therefore is natural to concentrate on other rings. There has been some attempt to use rings $A \cong \mathbb{Z}_4$, $D \cong \mathbb{Z}_2[u]/(u^2-u)$ and $G \cong \mathbb{Z}_2[u]/(u^2-1)$ in DNA codes. The other rings which have natural $GC$-content map did not get noticed. In particular, we construct DNA codes over the ring $E$.

\section{Quasi self-dual codes over $E$}

Following Alhmadi et al., a quasi self-dual code over the ring E is defined~\cite{Alahmadi1}. Recall the ring $E$ is defined by two generators $a$ and $b$ with the relation as follows.
\begin{center}
$E = \left\langle a,b \, | \, 2a=2b=0, a^2=a, b^2=b, ab=a, ba=b \right\rangle$.
\end{center}
Its multiplication table is given as follows.
\begin{table}[h]
\begin{center}
\begin{tabular}{ c | c | c | c | c | c |}
$\times$ & 0 & $a$ & $b$ & $c$ \\
\hline
0 & 0 & 0 & 0 & 0 \\
\hline
$a$ & 0 & $a$ & $a$ & 0 \\
\hline
$b$ & 0 & $b$ & $b$ & 0 \\
\hline
$c$ & 0 & $c$ & $c$ & 0 \\
\hline
\end{tabular}
\caption{Multiplication table of the ring of the ring $E$}
\end{center}
\end{table}

Since the ring $E$ is noncommutative, we first should define a linear $E$-code. A linear $E$-code is a one-sided $E$-submodule of $E^n$. Define an inner product on $E^n$ as $(x,y) = \sum_{i=1}^n x_i y_i$ where $x,y \in E^n$. (The product is the multiplication on $E$.)
\begin{defn}[\cite{Alahmadi1}]
{\rm
Let $\mathcal{C}$ be a linear $E$-code.
\begin{enumerate}
\item The {\em right dual} $\mathcal{C}^{\perp_R}$ of $\mathcal{C}$ is the right module $\mathcal{C}^{\perp_R} = \{ y \in E^n  \, | \, \forall x \in \mathcal{C}, (x,y) = 0\}$.

\item The {\em left dual} $\mathcal{C}^{\perp_L}$ of $\mathcal{C}$ is the left module $\mathcal{C}^{\perp_L} = \{ y \in E^n  \, | \, \forall x \in \mathcal{C}, (y,x) = 0\}$.

\item The code $\mathcal{C}$ is {\em left self-dual} (resp. {\em right self-dual}) if $\mathcal{C} = \mathcal{C}^{\perp_L}$ (resp. $\mathcal{C} = \mathcal{C}^{\perp_R}$). The code $\mathcal{C}$ is self-dual if it is both left and right self-dual.

\item The code $\mathcal{C}$ is {\em self-orthogonal} if $\forall x,y \in \mathcal{C}, (x,y) = 0$. A {\em quasi self-dual (QSD) code} is a self-orthogonal code of size $2^n$.
\end{enumerate}
}
\end{defn}

It is local with maximal ideal $J=\{0,c\}$, and its residue field $E/J \cong GF(2)$. Thus for any element $e\in E$, we can write
\begin{equation*}
e = as+ct
\end{equation*}
where $s,t \in \{0,1\} = GF(2)$ and where a natural action of $GF(2)$ on $E$. Denote by $r : E \to E/J \cong GF(2)$, the map of reduction modulo $J$. Thus $r(0)=r(c)=0$ and $r(a)=r(b)=1$. Then this map $r$ can be the $GC$-content map $\psi$. Let define $f : \{A,T,G,C\} \to E$ by

\[f(A)=0, f(T)=c, f(G)=a {\mbox{ and }} f(C)=b.\] Then $r(f(C))=r(f(G))=1$, and the others go to 0. We can check $\psi(x) = ax$ satisfies $0^C=c$ and $\psi(0)=\psi(c)=0$, $\psi(a)=\psi(b)=a$. Moreover $Im(\psi) = \{ 0,a \} \cong GF(2)$. Therefore $\psi \cong r$. So now let $\psi=r$. Then it can be extended from $E^n$ to ${GF(2)}^n$. And since pairing is given by $\{A,T\} \to \{0,c\}$ and $\{G,C\} \to \{a,b\}$, we should define $x^C = x+c$.

\begin{defn}[\cite{Alahmadi1}]
{\rm
Let $\mathcal{\mathcal{C}}$ be a code of length $n$ over $E$.
\begin{enumerate}
\item The {\em residue code} of $\mathcal{C}$ is $res(\mathcal{C}) = \{ \psi(y) \in {GF(2)}^n \, | \, y \in \mathcal{C}\}$.

\item The {\em torsion code} of $\mathcal{C}$ is $tor(\mathcal{C}) = \{ x \in {GF(2)}^n \, | \, cx \in \mathcal{C}\}$.
\end{enumerate}
}
\end{defn}

The both codes are binary code.

\begin{theorem}[\cite{Alahmadi1}]\label{Thm-3}
If $\mathcal{C}$ is a QSD code over $E$, then $C = a\, res(\mathcal{C}) \oplus c\, tor(\mathcal{C})$ as modules.
\end{theorem}

\begin{theorem}[\cite{Alahmadi1}] \label{Thm-4}
For any QSD $E$-linear codes $\mathcal{C}$, we have
\begin{enumerate}
\item $res(\mathcal{C}) \subseteq {res(\mathcal{C})}^{\perp}$,

\item $tor(\mathcal{C}) = {res(\mathcal{C})}^{\perp}$,

\item $| \, \mathcal{C} \, | = dim(res(\mathcal{C})) + dim(tor(\mathcal{C}))$.
\end{enumerate}
\end{theorem}

We can construct QSD $E$-codes by the above theorems.

\begin{theorem}[\cite{Alahmadi1}]
Let $\mathcal B$ be a self-orthogonal binary $\left[n,k_1\right]$ code. The code $\mathcal{C}$ over the ring $E$ defined by the relation
\begin{equation*}
\mathcal{C}=a \mathcal B \oplus c \mathcal B^{\perp}
\end{equation*}
is a QSD code. Its residue code is $\mathcal B$ and its torsion code is $\mathcal B^{\perp}$.
\end{theorem}

By above theorem, we know that the classification of QSD $E$-codes is equivalent to the classification of their residue codes. Moreover, two QSD codes $\mathcal{C}_1$ and $\mathcal{C}_2$ are equivalent up to permutation if and only if their residue codes are equivalent up to permutation. Therefore we can get the following theorem.

\begin{cor}
Let $N(n,k_1)$ be the number of inequivalent QSD codes over $E$ where $n$ is the length and $k_1$ is the dimension of their residue codes. Then
\begin{equation*}
N(n,k_1) = \Psi(n,k_1)
\end{equation*}
where $\Psi(n,k_1)$ is the number of inequivalent binary self-orthogonal codes.
\end{cor}

\begin{proof}

(i) Let $\mathcal B_1$ and $\mathcal B_2$ be self-orthogonal binary $[n,k_1]$ codes and let $\mathcal{C}_1$ and $\mathcal{C}_2$ be codes over the ring $E$ defined by $\mathcal{C}_i=a \mathcal B_i\oplus c \mathcal B_i^\perp, i=1,2$.

By theorem 5, if $\mathcal B_1 \cong \mathcal B_2$ then the code $a \mathcal B_1\oplus c \mathcal B_1^\perp \cong a \mathcal B_2\oplus c \mathcal B_2^\perp$ so $\mathcal{C}_1 \cong \mathcal{C}_2$.

(ii)
Suppose $\mathcal{C}_1 \cong \mathcal{C}_2$. Then by Theorems~\ref{Thm-3} and~\ref{Thm-4}, $a \mathcal B_1\oplus c \mathcal B_1^\perp = \mathcal C_1$  for some binary code $\mathcal B_1$ and $a \mathcal B_2\oplus c \mathcal B_2^\perp =\mathcal C_2$ for some binary code $\mathcal B_2$. Therefore, $a\mathcal B_1 \cong a\mathcal B_2$, that is, $\mathcal B_1 \cong \mathcal B_2$.

Therefore if there exist $\Psi(n,k_1)$ inequivalent self-orthogonal binary $[n,k_1]$ codes, then $\Psi(n,k_1)$ is equal to the number of inequivalent QSD $E$-codes $\mathcal C$ which have length $n$ and $dim(res(\mathcal(\mathcal{C})))=k_1$.
\end{proof}

We need the classification of inequivalent binary self-orthogonal codes. Hou et al., classified the case when $k \le 5$ and $n \le 40$~\cite{Hou}. Pless classified the case when $k = n/2$ and $n \le 20$~\cite{Pless}.

\begin{lemma}
The number of inequivalent binary self-orthogonal $[n,k,2]$ codes is the number of inequivalent binary self-orthogonal $[n-2,k-1]$ codes.
\end{lemma}

\begin{proof}
Let $\mathcal{C}$ be a binary self-orthogonal $[n,k,2]$ code. Let take $x \in \mathcal{C}$ with $wt(x)=2$. Then we can puncture the positions of nonzero coordinates of $x$. We can get self-orthogonal $[n-2,k-1]$ code. Conversely, we can get $[n,k,2]$ codes from $[n-2,k-1]$ codes by adding weight 2 extra vector. Obviously if two codes are equivalent, then the induced codes are also equivalent.
\end{proof}

We can calculate the number of self-orthogonal $[14,6]$ codes and self-orthogonal $[15,6]$ codes using the lemma and the paper of I. Bouyukliev~\cite{Bouyukliev}.

\begin{lemma}
The number of inequivalent binary self-orthogonal $[14,6]$ codes is 27. The number of inequivalent binary self-orthogonal $[15,6]$ codes is 48.
\end{lemma}

\begin{proof}
Note that the largest minimum distance of binary $[14,6]$ codes is 5. So we can let $d\le5$ where $d$ is the minimum distance. Since the binary self orthogonal codes have only even weights, we consider the cases $d=2$ and $d=4$. By the above lemma, the number of self-orthogonal $[14,6,2]$ codes is the number of self-orthogonal $[12,5]$ codes, that is 15. By the paper of I. Bouyukliev, there exist twelve self-orthogonal $[14,6,4]$ codes~\cite{Bouyukliev}. Hence the number of self-orthogonal $[14,6]$ codes is $15+12=27$.

Similarly we need to compute the numbers of self-orthogonal $[15,6,2]$ codes and $[15,6,4]$ codes. They are 23 and 25, respectively. Even though there exist a $[15,6,6]$ code, it cannot be self-orthogonal. So the number of $[15,6]$ codes is $23+25=48$.
\end{proof}

So we can finish the following table.

\vskip 3mm
\resizebox{\textwidth}{!}{
\begin{tabular}{|c|c|c|c|c|c|c|c|c|c|c|c|c|c|c|c|c|c|}
\hline
$n$ & $1$ & \multicolumn{2}{|c|} {2} & \multicolumn{2}{|c|} {3} & \multicolumn{3}{|c|} {4} & \multicolumn{3}{|c|} {5} & \multicolumn{4}{|c|} {6} & \multicolumn{2}{|c|} 7 \\
\hline
$k_1$ & 0 & 0 & 1 & 0 & 1 & 0 & 1 & 2 & 0 & 1 & 2 & 0 & 1 & 2 & 3 & 0 & 1 \\
\hline
$N$ & 1 & 1 & 1 & 1 & 1 & 1 & 2 & 1 & 1 & 2 & 1 & 1 & 3 & 3 & 1 & 1 & 3 \\
\hline
\end{tabular}
}
\vskip 3mm
\resizebox{\textwidth}{!}{
\begin{tabular}{|c|c|c|c|c|c|c|c|c|c|c|c|c|c|c|c|c|c|}
\hline
\multicolumn{2}{|c|} {7}& \multicolumn{5}{|c|} {8} & \multicolumn{5}{|c|} {9} & \multicolumn{6}{|c|} {10} \\
\hline
2 & 3 & 0 & 1 & 2 & 3 & 4 & 0 & 1 & 2 & 3 & 4 & 0 & 1 & 2 & 3 & 4 & 5 \\
\hline
3 & 2 & 1 & 4 & 6 & 5 & 2 & 1 & 4 & 6 & 6 & 3 & 1 & 5 & 10 & 12 & 9 & 2 \\
\hline
\end{tabular}
}
\vskip 3mm
\resizebox{\textwidth}{!}{
\begin{tabular}{|c|c|c|c|c|c|c|c|c|c|c|c|c|c|c|c|c|c|}
\hline
\multicolumn{6}{|c|} {11} & \multicolumn{7}{|c|} {12} & \multicolumn{5}{|c|} {13} \\
\hline
0 & 1 & 2 & 3 & 4 & 5 & 0 & 1 & 2 & 3 & 4 & 5 & 6 & 0 & 1 & 2 & 3 & 4 \\
\hline
1 & 5 & 10 & 14 & 12 & 4 & 1 & 6 & 16 & 26 & 28 & 15 & 3 & 1 & 6 & 16 & 30 & 36 \\
\hline
\end{tabular}
}
\vskip 3mm
\resizebox{\textwidth}{!}{
\begin{tabular}{|c|c|c|c|c|c|c|c|c|c|c|c|c|c|c|c|c|c|}
\hline
\multicolumn{2}{|c|} {13} & \multicolumn{8}{|c|} {14} & \multicolumn{8}{|c|} {15} \\
\hline
5 & 6 & 0 & 1 & 2 & 3 & 4 & 5 & 6 & 7 & 0 & 1 & 2 & 3 & 4 & 5 & 6 & 7 \\
\hline
23 & 6 & 1 & 7 & 23 & 51 & 75 & 61 & 27 & 4 & 1 & 7 & 23 & 58 & 98 & 94 & 48 & 10 \\
\hline
\end{tabular}
}
\vskip 6mm

\section{Quasi self-dual DNA codes over $E$}

\begin{theorem}[\cite{Alahmadi1}]
Let $\mathcal{C}$ be a QSD code over $E$. Then
\begin{equation*}
CWE_\mathcal{C}(w,x,y,z) = J(res(\mathcal{C}),tor(\mathcal{C}))(w,x,y,z)
\end{equation*}
where $J(A,B)$ of two binary linear codes $A,B$ is the joint weight enumerator defined by
\begin{equation*}
J(A,B)(w,x,y,z) \sum_{u\in A, v \in B} w^{i(u,v)} x^{j(u,v)} y^{k(u,v)} z^{l(u,v)},
\end{equation*}
$i,j,k,l$ are the integers of the number of indices $\iota \in \{1,\cdots,n\}$ with $(u_\iota,v_\iota)=(0,0),(0,1),(1,0)$ and $(1,1)$, respectively.
\end{theorem}

\begin{theorem} \label{thm-GC}
Let $\mathcal{C}$ be a QSD code over $E$. Then
\begin{equation*}
GCW_\mathcal{C}(x,y) = \displaystyle \sum_{i=0}^n 2^{n-k_1} A_i(res(C)) x^i y^{n-i}
\end{equation*}
where $n = | \, \mathcal{C} \, |, k_1=dim(res(\mathcal{C}))$ and $A_i(res(\mathcal{C}))$ is the binary weight distribution of $res(\mathcal{C})$.
\end{theorem}

\begin{proof}
By MacWilliams~\cite{MacWilliams},
\begin{equation*}
W_{res(\mathcal{C})}(x,y) = \frac{1}{dim(tor(\mathcal{C}))} J(res(\mathcal{C}),tor(\mathcal{C}))(x,x,y,y).
\end{equation*}
So
\begin{equation*}
\begin{aligned}
GCW_\mathcal{C}(x,y)=CWE_\mathcal{C}(x,x,y,y)=J(res(\mathcal{C}),tor(\mathcal{C}))(x,x,y,y) \\ = dim(tor(\mathcal{C})) W_{res(\mathcal{C})}(x,y) = 2^{n-k_1} W_{res(\mathcal{C})}(x,y).
\end{aligned}
\end{equation*}
\end{proof}

For example, let $res(\mathcal{C}) = \left\langle (1 \ 1 \ \cdots \ 1 \ 0 \ 0 \cdots \ 0) \right\rangle$, where the codeword has $m$ ones and $n-m$ zeros ($m$ is even). Then $res(\mathcal{C})$ is a 1-dimensional code, therefore
\begin{equation*}
GCW_\mathcal{C}(x,y)= 2^{n-1} x^m y^{n-m} + 2^{n-1} y^n.
\end{equation*}
We can get $(n,2^{n-1})$ DNA codes which have fixed $GC$-content constraint $m$ and 0 (resp).

We want to give reverse (and reverse-complement) constraints for QSD DNA codes. Since any residue code has zero vector, its QSD DNA code has the vector $(cc\cdots c)$. Since the complement map is defined $x^C = x + c$ in the ring $E$, $x^{C} \in \mathcal{C}$ for any $x \in \mathcal{C}$ and for any QSD DNA code $\mathcal{C}$. Then $\min\{d_H(x^{RC},y) \, | \, x,y \in \mathcal{C} \}=0$. Thus we need to give reverse-complement constraints to a subcode which have a fixed $GC$-content constraint.

\begin{defn}{\em
Let $\mathcal{C}$ be a QSD DNA code of length $n$. Let $\mathcal{C}_m$ be the subcode of $\mathcal{C}$, which has a fixed $GC$-content constraint $m$. This $\mathcal{C}_m$ has permutation equivalence codes, $\mathcal{P}_m = \{ \sigma(\mathcal{C}_m) \, | \, \sigma \in S_n \}$. Define
\begin{equation*}
d_{RC}^{\, m} := \max_{\mathcal{C}^{\prime} \in \mathcal{P}_m}\{ d_{\mathcal{C}^{\prime}} \}
\end{equation*}
where $d_{\mathcal{C}^{\prime}} = \min \{d_H(x^{RC},y) \, | \, x,y \in \mathcal{C}^{\prime} \}$.
}
\end{defn}

It is clear that $d_{RC}^{\, 0} = 0$ since the zero vector and $(cc\cdots c)$ in $\mathcal{C}_0$.

\begin{theorem}\label{lemma1}
Let $\mathcal{C}$ be a QSD code over $E$ satisfying $res(\mathcal{C}) = \left\langle a_1 \right\rangle = \left\langle (1 \ \cdots \ 1 \ 0 \cdots \ 0) \right\rangle$ where $d_H(a_1) = m$ ($m$ is even). Then $d_{RC}^{\, m} = 2\min\{m,n-m\}$.
\end{theorem}
\begin{proof}
Let take $\sigma_1, \sigma_2 \in S_n$. Suppose that
\begin{equation*}
d_H(\sigma_1(a_1),{\sigma_1(a_1)}^R) \le d_H(\sigma_2(a_1),{\sigma_2(a_1)}^R).
\end{equation*}
Denote $\sigma_1(a_1) = (x_1 \cdots x_n)$ where $x_i \in GF(2)$. Then $d_H(\sigma_1(a_1),{\sigma_1(a_1)}^R) =$ the number of $x_i's$, where $x_i \ne x_{n-i}$. We claim that $d_{\sigma_1(\mathcal{C}_m)} =  d_H(\sigma_1(a_1),{\sigma_1(a_1)}^R)$. If $x_j \ne x_{n-j}$ for some $j$, then $ax_j + ct_1 \ne ax_{n-j} + ct_2$ for any $t_1, t_2 \in GF(2)$ (since $a(x_j+x_{n-j}) = a \ne c(t_1+t_2)$). So $d_H(x^{RC},y) \ge 2$ for any $x,y \in \sigma_1(\mathcal{C}_m)$ (For any $x,y$, it generated by $\sigma_1(a_1)$ and so the counting appears in the $j$th and $(n-j)$th position). If there exist $k$ coordinates $j^{\prime}$s which satisfy $x_{j^{\prime}} \ne x_{n-j^{\prime}}$, then $d_H(x^{RC},y) \ge 2k$ for any $x,y \in \sigma_1(\mathcal{C}_m)$. Therefore $d_{\sigma_1(\mathcal{C}_m)} =  d_H(\sigma_1(a_1),{\sigma_1(a_1)}^R)$. This claim means that we can get the minimum distance of the subcode $\sigma_1(\mathcal{C}_m)$ using the distance of $\sigma_1(a_1)$.

Then by the assumption we can get that
\begin{equation*}
d_{\sigma_1(\mathcal{C}_m)} \le d_{\sigma_2(\mathcal{C}_m)}.
\end{equation*}
Therefore we need to increase the number of $x_i$'s satisfying $x_i \ne x_{n-i}$. If $m < n/2$, then we can take $\sigma \in S_n$ where $\sigma(a_1) = a_1 = (1 \ \cdots \ 1 \ 0 \cdots \ 0)$ so that there are $m$ positions of $x_i$'s satisfying $x_i \ne x_{n-i}$. Thus $d_{\sigma(\mathcal{C}_m)} = 2m$. If $m \ge n/2$, then also $\sigma(a_1) = a_1 = (1 \ \cdots \ 1 \ 0 \cdots \ 0)$ has $n-m$ positions of $x_i$'s satisfying $x_i \ne x_{n-i}$ so that $d_{\sigma(\mathcal{C}_m)} = 2(n-m)$. Therefore
\begin{equation*}
d_{RC}^m=
\begin{cases}
2m, & \mbox{if } m < n/2 \\
2(n-m) & \mbox{if } m \ge n/2
\end{cases}.
\end{equation*}
If $2m < 2(n-m)$, then $m < n/2$ so $d_{RC}^m = 2m$. If $2m \ge 2(n-m)$, then $m \ge n/2$ so $d_{RC}^m = 2(n-m)$. Thus $d_{RC}^{\, m} = 2\min\{m,n-m\}$.
\end{proof}

\begin{theorem}
Let $\mathcal{C}$ be a QSD code over $E$ satisfying
\begin{equation*}
res(\mathcal{C}) = \left\langle \begin{pmatrix} a_1 \\ a_2 \end{pmatrix} \right\rangle = \left\langle \begin{pmatrix} 1 \cdots 1 & 0 \cdots 0 & 0 \cdots 0 \\ 0 \cdots 0 & 1 \cdots 1 & 0 \cdots 0 \end{pmatrix} \right\rangle
\end{equation*}
where $d_H(a_1)=m_1$ and $d_H(a_2)=m_2$ ($m_1$ and $m_2$ are positive even integers). Let $m=m_1+m_2$. Then the following hold.
\begin{enumerate}
\item If $m_1=m_2$, then $d_{RC}^{\, m_1} = d_{RC}^{\, m_2} = \min\{m, 2(n - \lfloor n/2 \rfloor) - m \}$ and $d_{RC}^{\, m} = 2\min\{m,n-m \}$.
\item If $m_1\ne m_2$, then $d_{RC}^{\, m_1} = 2\min\{m_1,n-m_1 \}$, $d_{RC}^{\, m_2} = 2\min\{m_2,n-m_2 \}$ and $d_{RC}^{\, m} = 2\min\{m,n-m \}$.
\end{enumerate}
Note that $2(n - \lfloor n/2 \rfloor) - m = \begin{cases}
n-m, & \mbox{if } n \mbox{ is even} \\
n-m+1, & \mbox{if } n \mbox{ is odd}
\end{cases}$.
\end{theorem}
\begin{proof}
- Case 1. Suppose $m_1 = m_2 = m/2$. $\mathcal{C}_m$ is generated by one vector $(a_1+a_2)$, so by Theorem \ref{lemma1}, $d_{RC}^{\, m} = 2\min\{m,n-m \}$. Assume that $m < n/2$. Then the codewords generated by $(a_1)$ or $(a_2)$ can have $2m_1 = m$ positions of $x_i$'s satisfying $x_i \ne x_{n-i}$. Thus $d_{RC}^{\, m_1} = 2m_1 = m$.

Now assume that $n/2 \le m$. Let $\sigma_1(a_1) = (1 \ \cdots \ 1 \ 0 \cdots \ 0)$. By the assumption $\sigma_1(a_2)$ has to form that $\sigma_1(a_2) = (0 \ \cdots \ 0 \ x_{m_1+1} \cdots \ x_n)$ where $x_i \in GF(2)$.

Let $n$ be even. To avoid the coincidence, we should take $x_i$'s such that $x_{m_1+1} = \cdots = x_{n/2}=1$. Locate the rest of ones $x_{n-2m_1+ n/2 +1} = \cdots = x_n = 1$. Then $d_H(\sigma_1(a_1),({\sigma_1(a_2)}^{R})) = 2(n/2 - m_1) = n-m$, $d_H(\sigma_1(a_1),({\sigma_1(a_1)}^{R})) = d_H(\sigma_1(a_2),({\sigma_1(a_2)}^{R})) = 2m_1 = m$. Since $m \le n$, so $d_{RC}^{\, m_1} = n-m$.

Next, let $n$ be odd. If we take the same progress as the $n$ even case, we can get $d_H(\sigma_1(a_1),({\sigma_1(a_2)}^{R})) = 2(\lfloor n/2 \rfloor - m_1) = 2\lfloor n/2 \rfloor - m$. However, since $n$ is odd, we can let $x_{\lfloor n/2 \rfloor +1} = 1$, which is in $\sigma_1(a_2)$. In that case, $d_H(\sigma_1(a_1),({\sigma_1(a_2)}^{R})) = 2\lfloor n/2 \rfloor - 2m + 2$, $d_H(\sigma_1(a_1),({\sigma_1(a_1)}^{R})) = m$, $d_H(\sigma_1(a_2),({\sigma_1(a_2)}^{R})) = m-2$. Since $n/2 < m$, so $\lfloor n/2 \rfloor + 1\le m$. Then $2\lfloor n/2 \rfloor + 4 \le 2m+2 \le 3m$ (since $2 \le m$).

Therefore $d_{RC}^{\, m_1} = 2\lfloor n/2 \rfloor - m + 2$. Then $d_{RC}^{\, m_1}$ can be formed as $d_{RC}^{\, m_1} = 2(n - \lfloor n/2 \rfloor) - m$. Thus
\begin{equation*}
d_{RC}^{m_1}=
\begin{cases}
m, & \mbox{if } m < n/2 \\
2(n - \lfloor n/2 \rfloor) - m, & \mbox{if } m \ge n/2
\end{cases}
\end{equation*}

If $m < n/2$, then $m < n-m \le 2(n - \lfloor n/2 \rfloor) - m$. Thus $d_{RC}^{\, m_1} = \min\{m, 2(n - \lfloor n/2 \rfloor) - m \}$.

- Case 2. Suppose $m_1 \ne m_2$. Then the subcode with fixed $GC$-content constraint $m_1$ is generated by one vector $a_1$. So by Theorem \ref{lemma1}, $d_{RC}^{\, m_1} = 2\min\{m_1,n-m_1 \}$. In the same argument, we can get the following results: $d_{RC}^{\, m_2} = 2\min\{m_2,n-m_2 \}$, and $\mathcal{C}_m$ is generated by one vector $(a_1+a_2)$, so $d_{RC}^{\, m} = 2\min\{m,n-m \}$.
\end{proof}

\begin{theorem}
Let $\mathcal{C}$ be a QSD code over $E$ satisfying
\begin{equation*}
res(\mathcal{C}) = \left\langle \begin{pmatrix} a_1 \\ a_2 \end{pmatrix} \right\rangle = \left\langle \begin{pmatrix} 1 \cdots 1 & 0 \cdots 0 & 1 \cdots 1 & 0 \cdots 0 \\ 0 \cdots 0 & 1 \cdots 1 & 1 \cdots 1 & 0 \cdots 0 \end{pmatrix} \right\rangle
\end{equation*}
where  $d_H(a_1)=m_1+m_3$, $d_H(a_2)=m_2+m_3$ and $d_H(a_1 \cap a_2) = m_3$ ($m_1$, $m_2$ and $m_3$ are positive even integers). Then the following holds.
\begin{enumerate}
\item If $m_1,m_2$ and $m_3$ are all distinct, then $d_{RC}^{\, m_i+m_j} = 2\min\{m_i+m_j,n-(m_i+m_j) \}$ for all $1 \le i \ne j \le 3$.

\item Without loss of generality, let $m_1=m_2\ne m_3$. Then $d_{RC}^{\, m_1+m_2} = d_{RC}^{\, 2m_1}  = 2\min\{2m_1,n-2m_1 \}$ and $d_{RC}^{\, m_1+m_3} = d_{RC}^{\, m_2+m_3}$ is
\begin{equation*}
d_{RC}^{\, m_1+m_3} = \begin{cases}
2(m_1+m_3) & \mbox{if } 2m_1+m_3 < n/2 \\
n-2m_1-\delta_n & \mbox{if } n/2 \le 2m_1+m_3 < n/2+m_3 \\
2(\lfloor n/2 \rfloor-m_1) & \mbox{if } n/2 \le 2m_1
\end{cases}.
\end{equation*}

\item If $m_1 = m_2 = m_3$, then $d_{RC}^{\, m_1+m_2} = d_{RC}^{\, m_2+m_3} = d_{RC}^{\, m_1+m_3}$ is
\begin{equation*}
d_{RC}^{\, m_1+m_2} = \begin{cases}
4m_1 & \mbox{if } m_1 < n/6 \\
n-2m_1-\delta_n & \mbox{if } n/6 \le m_1 < n/4 \\
2(\lfloor n/2 \rfloor-m_1) & \mbox{if } n/4 \le m_1
\end{cases}.
\end{equation*}
\end{enumerate}
where $\delta_n = \begin{cases}
0 & \mbox{if } n \equiv 0 \mod 4 \\
1 & \mbox{if } n \equiv 1 \mod 2 \\
2 & \mbox{if } n \equiv 2 \mod 4
\end{cases}$.
\end{theorem}
\begin{proof}
- Case 1. If $m_1,m_2$ and $m_3$ are all distinct, then the subcodes which have fixed $GC$-content constraint are generated by one vector. So $d_{RC}^{\, m_i+m_j} = 2\min\{m_i+m_j,n-(m_i+m_j) \}$ is obvious.

- Case 2. Suppose $m_1 \ne m_2 = m_3$. Then this case is obviously the same as the case $m_1 = m_3 \ne m_2$. And let $a_3 = a_1+a_2$. Then $res(\mathcal{C})$ can be generated by $a_1$ and $a_3$. Since $a_3 = a_1+a_2$, so $d_H(a_3) = d_H(a_1) + d_H(a_2) - 2d_H(a_1 \cap a_2) = m_1+m_2$ and $d_H(a_1 \cap a_3) = m_1$. Therefore the case $m_1 \ne m_2 = m_3$ is the same as the case $m_1=m_2\ne m_3$.

So now suppose that $m_1=m_2\ne m_3$. Then only the code $a_1+a_2$ generates the codeword which has fixed $m_1+m_2$ $GC$-content. Thus $d_{RC}^{\, 2m_1} = 2\min\{2m_1,n-2m_1 \}$ is obvious.

If $m_1+m_2+m_3 = 2m_1+m_3 \le n/2$, then we can easily check that $d_{RC}^{\, m_1+m_3} = 2\min\{m_1+m_3,n-(m_1+m_3) \}$. Note that $2m_1+m_3 \le n/2$ implies that $2\min\{m_1+m_3,n-(m_1+m_3) \} = 2(m_1+m_3)$. So $d_{RC}^{\, m_1+m_3} = 2(m_1+m_3)$.

Assume that $m_1+m_2+m_3 = 2m_1+m_3 > n/2$ and $m_1+m_2 = 2m_1 < n/2$. Denote $\sigma_1(a_1) = (x_1 \cdots x_n)$ and $\sigma_2(a_1) = (y_1 \cdots y_n)$ where $x_i, y_i \in GF(2)$. Then we can let $x_1 = \cdots x_{m_1} = 1$, $x_{m_1+1} = \cdots = x_{2m_1} = 0$, $y_1  = \cdots = y_{m_1} = 0$, $y_{m_1+1} = \cdots = y_{2m_1} = 1$.

Now let consider $n \equiv 0 \mod 4$. To avoid the coincidence we should let $x_{2m_1+1} = \cdots = x_{n/2} = 1 = y_{2m_1+1} = \cdots = y_{n/2}$. Then rest 1's should be located in $x_{n/2+1}, \ldots, x_n$ and $y_{n/2+1}, \ldots, y_n$. Since $n \equiv 0 \mod 4$, the number of rest 1's is $m_3 - (n/2 - 2m_1)$ so it is even. If we let 1's to one side, the coincidence will be increasing. Thus we can take $x_{n-3m_1-m_3/2+n/4+1} = \cdots = x_{n-m_1+m_3/2-n/4} = 1$. Then the number of 1's is $m_3 - n/2 + 2m_1$ and the middle point is between $n-2m_1$ and $n-2m_1+1$. In this case $d_H(\sigma_1(a_1), {\sigma_1(a_2)}^R) = n-2m_1$. The other Hamming distance is not smaller than $n-2m_1$.

If $n \equiv 2 \mod 4$, then the number of 1's $m_3 - n/2 + 2m_1$ is not even so we cannot divide into half. So one side has more 1's, and then the minimum distance value is decreasing exactly 2. If $n \equiv 1 \mod 2$, then $x_{\lfloor n/2 \rfloor +1} = y_{\lfloor n/2 \rfloor +1} = 1$. Then the minimum distance value is decreasing exactly 1. Therefore $d_{RC}^{\, m_1+m_2} = n-2m_1-\delta_n$.

Lastly assume that $n/2 \le m_1+m_2 = 2m_1$. Denote $\sigma_1(a_1) = (x_1 \cdots x_n)$ and $\sigma_2(a_1) = (y_1 \cdots y_n)$. Let $x_1 = \cdots = x_{m_1} = 1$, $y_{\lfloor n/2 \rfloor + 1} = \cdots = y_{\lfloor n/2 \rfloor + m_1} = 1$. And let $x_{m_1+1} = \cdots = x_{m_1+m_3/2} = 1$, $y_{\lfloor n/2 \rfloor + m_1 + 1} = \cdots = y_{\lfloor n/2 \rfloor + m_1 = m_3/2} = 1$. Then $d_{RC}^{\, m_1+m_2} = 2 \times (m_3/2) + 2 \times (\lfloor n/2 \rfloor- m_1 -m_3/2) = 2(\lfloor n/2 \rfloor-m_1)$.

- Case 3. Assume that $m_1=m_2=m_3$. Then we can apply the same methodas the case 2.
\end{proof}

For example, let $n=5$ and $k=2$. It is easy to see that there is a unique binary self-orthogonal $[5,2,3]$ code $\mathcal B$ with generator matrix
$
 \begin{pmatrix}
1 & 1 & 0 & 0 & 0 \\
0 & 0 & 1 & 1 & 0 \\
\end{pmatrix}.
$
 This gives the residue code $res(\mathcal{C})= \left\langle \begin{pmatrix} 1 & 1 & 0 & 0 & 0 \\ 0 & 0 & 1 & 1 & 0 \end{pmatrix} \right\rangle$.
The dual of $\mathcal B$ is generated by
$
 \begin{pmatrix}
1 & 1 & 0 & 0 & 0 \\
0 & 0 & 1 & 1 & 0 \\
0 & 0 & 0 & 0 & 1 \\
\end{pmatrix}.
$
Thus the generator matrix for a QSD code $\mathcal C$ is
\[
 \begin{pmatrix}
a & a & 0 & 0 & 0 \\
0 & 0 & a & a & 0 \\
0 & 0 & 0 & 0 & c \\
\end{pmatrix}.
\]

 Then by the formula (Theorem 10), $d_{RC}^2=2$ and $d_{RC}^4=2$. See the table at the end of the paper.

 In general, we can calculate some $d_{RC}$ values. The table in the conclusion shows some proper value of $d_{RC}$ for each length and dimension of the residue codes. The tables of specific $d_{RC}$ values up to the classification of QSD DNA codes with $n \le 8$ are in the conclusion. The  Magma source code is available at J.-L. Kim's website~\cite{KimWeb}.

\section{Quasi self-dual DNA codes over $F$}

The ring $F$ is defined by
\begin{equation*}
F = \left\langle a,b \, | \, 2a=2b=0, a^2=a, b^2=b, ab=b, ba=a \right\rangle.
\end{equation*}
Thus its multiplication table is given as follows.
\begin{table}[h]
\begin{center}
\begin{tabular}{ c | c | c | c | c | c |}
$\times$ & 0 & $a$ & $b$ & $c$ \\
\hline
0 & 0 & 0 & 0 & 0 \\
\hline
$a$ & 0 & $a$ & $b$ & $c$ \\
\hline
$b$ & 0 & $a$ & $b$ & $c$ \\
\hline
$c$ & 0 & 0 & 0 & 0 \\
\hline
\end{tabular}
\caption{Multiplication table of the ring of the ring $F$}
\end{center}
\end{table}

The ring $E$ and $F$ are not isomorphic. Even though ${(x,y)}_E \ne {(x,y)}_F$ for inner products, we can define a QSD DNA code over the ring $F$ similarly. Let a linear $F$-code be a one-sided $F$-submodule of $F^n$.

\begin{defn}{\em
Let $x,y \in F^n$ where $x = (x_1,\cdots,x_n)$ and $y = (y_1,\cdots,y_n)$. Define an inner product of $x,y$ as $(x,y) = \sum x_iy_i$. Let $\mathcal{C}$ be a linear $F$-code.
\begin{enumerate}
\item The {\em right dual} $\mathcal{C}^{\perp_R}$ of $\mathcal{C}$ is the right module $\mathcal{C}^{\perp_R} = \{ y \in F^n  \, | \, \forall x \in \mathcal{C}, (x,y) = 0\}$.

\item The {\em left dual} $\mathcal{C}^{\perp_L}$ of $\mathcal{C}$ is the left module $\mathcal{C}^{\perp_L} = \{ y \in F^n  \, | \, \forall x \in \mathcal{C}, (y,x) = 0\}$.

\item The code $\mathcal{C}$ is {\em left self-dual} (resp. {\em right self-dual}) if $\mathcal{C} = \mathcal{C}^{\perp_L}$ (resp. $\mathcal{C} = \mathcal{C}^{\perp_R}$). And the code $\mathcal{C}$ is {\em self-dual} is it is both of its duals.

\item The code $\mathcal{C}$ is {\em self-orthogonal} if $\forall x,y \in \mathcal{C}, (x,y) = 0$. A {\em quasi self-dual (QSD) code} is a self-orthogonal code of size $2^n$.
\end{enumerate}
}
\end{defn}

Remark that ${(x,y)}_E \ne {(x,y)}_F$ as an inner product. However if $\mathcal{C}$ is QSD in the ring $E$, then so is in the ring $F$.

\begin{theorem}
Let $\mathcal{C}$ be a QSD code over the ring $E$. Then by a map $f:E \mapsto F$, $f(\mathcal{C})$ is a QSD code over ring $F$.
\end{theorem}

\begin{proof}
Define a bijection $f:E \mapsto F$ by $f(a_E) = a_F$, $f(b_E) = b_F$ and $f(c_E) = c_F$. Let $\mathcal{C}$ be a QSD code over the ring $E$. Take $x,y \in \mathcal{C}$, denoted by $x=(x_1,\ldots,x_n)$ and $y=(y_1,\ldots,y_n)$. Then ${(x,y)}_E = \sum_{i=1}^n (x_i,y_i) =0$ implies that
\begin{center}
$\displaystyle \sum_{x_m,y_m = c} (x_m,y_m) +\sum_{x_{m_1} \ne 0 \mbox{ nor } c} (x_{m_1}, c) + \sum_{y_{n_1} \ne 0 \mbox{ nor } c} (c, y_{n_1}) + \sum_{\substack{x_{m_2} \ne 0 \mbox{ nor } c \\ y_{n_2} \ne 0 \mbox{ nor } c}} (x_{m_2}, y_{n_2})$

$\displaystyle = \sum_{y_{n_1} \ne 0 \mbox{ nor } c} (c, y_{n_1}) + \sum_{\substack{x_{m_2} \ne 0 \mbox{ nor } c \\ y_{n_2} \ne 0 \mbox{ nor } c}} (x_{m_2}, y_{n_2}) = 0$
\end{center}
Note that there are only $a$ terms or $b$ terms in the summation
\begin{equation*}
sum_{\substack{x_{m_2} \ne 0 \mbox{ nor } c \\ y_{n_2} \ne 0 \mbox{ nor } c}} (x_{m_2}, y_{n_2}).
\end{equation*}
Therefore the number of coordinates of $a$'s and $b$'s in $y$ is even. If the number of $a$'s in odd, then $(y,y)\ne 0$. So both the number of $a$'s and the number of $b$'s are even. Hence every codeword in $C$ has even $a$-positions and $b$-positions. Then
\begin{equation*}
\sum_{i=1}^n (f(x_i),f(y_i)) = 0
\end{equation*}
since
\begin{equation*}
\sum_{x_{m} \ne 0 \mbox{ nor } c} (f(x_{m}), f(c)) + \sum_{\substack{x_{m} \ne 0, c \\ y_{n} \ne 0 \mbox{ nor } c}} (f(x_{m}), f(y_{n})) = 0
\end{equation*}.
\end{proof}

So we can regard an QSD code over the ring $E$ as an QSD code over the ring $F$.

\begin{defn}{\em
Let $\mathcal{C}$ be a code of length $n$ over $F$.
\begin{enumerate}
\item The {\em residue code} of $\mathcal{C}$ is $res(\mathcal{C}) = \{ \psi(y) \, | \, y \in \mathcal{C}\}$.

\item The {\em torsion code} of $\mathcal{C}$ is $tor(\mathcal{C}) = \{ x \in {GF(2)}^n \, | \, cx \in \mathcal{C}\}$
\end{enumerate}
where $\psi:F \to GF(2)$ is the map $\psi(0)=\psi(c)=0$ and $\psi(a)=\psi(b)=1$, or $\psi(x) = xa$. The map $\psi(x)=xa$ has an image $\{0,c\} \cong GF(2)$ so this map $\psi$ is well-defined.
}
\end{defn}

\begin{lemma}
Every element $f \in F$ can be written
\begin{equation*}
f = as + ct
\end{equation*}
where $s,t \in GF(2)$.
\end{lemma}
Since the ring $F$ is isomorphic to the ring $E$ as additive group, $F$ also has this decomposition.

\begin{cor}
If $\mathcal{C}$ is a QSD code over $F$, then $C = a\, res(\mathcal{C}) \oplus c\, tor(\mathcal{C})$ as modules.
\end{cor}

\begin{cor}
Let $N(n,k_1)$ be the number of inequivalent QSD codes over $F$ where $n$ is the length and $k_1$ is the dimension of their residue codes. Then
\begin{equation*}
N(n,k_1) = \Psi(n,k_1)
\end{equation*}
where $\Psi(n,k_1)$ is the number of inequivalent binary self-orthogonal codes.
\end{cor}

\begin{cor}
Let $\mathcal{C}$ be a QSD code over $F$. Then
\begin{equation*}
GCW_\mathcal{C}(x,y) = \sum_{i=0}^n 2^{n-k_1} A_i(res(\mathcal{C})) x^i y^{n-i}
\end{equation*}
where $n = | \, \mathcal{C} \, |, k_1=dim(res(\mathcal{C}))$ and $A_i(res(\mathcal{C}))$ is the binary weight distribution of $res(C)$.
\end{cor}

Therefore we can check $F$ has the same $GC$-weight distribution over $E$.

\section{Conclusion}

In this paper, we construct QSD DNA codes over $E$. For each DNA code, the $GC$-weight enumerator is obtained. This implies the (nonlinear) subcodes which have a fixed $GC$-content. Especially some minimum distances with reverse complement constraint in the ring $E$ are calculated for $n\le 8$. The tables of $d_{RC}$ are below. Some values of $d_{RC}$ are computed by MAGMA programming. The QSD DNA codes over the ring $F$ is almost same as the case of the ring $E$, so we can apply the below tables.

\begin{table}
\begin{adjustbox}{angle=90}
\resizebox{7.2cm}{!}{
\begin{tabular}{ | c | c | c | c | c | }
\hline
$n$ & $k_1$ & \centering Residue code & \centering Generator Matrix & $d_{RC}^m$ \\
\hline
$n$ & 0 & \centering $\left\langle \begin{pmatrix} 0 & \cdots & 0\end{pmatrix} \right\rangle$ & \centering $\begin{pmatrix} c & \cdots & c\end{pmatrix}$ & $d_{RC}^0=0$ \\
\hline
2 & 1 & \centering $\left\langle \begin{pmatrix} 1 & 1 \end{pmatrix} \right\rangle$ & \centering $\begin{pmatrix} a & a \end{pmatrix}$ & $d_{RC}^2=0$ \\
\hline
3 & 1 & \centering $\left\langle \begin{pmatrix} 1 & 1 & 0 \end{pmatrix} \right\rangle$ & \centering $\begin{pmatrix} a & a & 0 \\ 0 & 0 & c \end{pmatrix}$ & $d_{RC}^2=2$ \\
\hline
4 & 1 & \centering $\left\langle \begin{pmatrix} 1 & 1 & 0 & 0 \end{pmatrix} \right\rangle$ & \centering $\begin{pmatrix} a & a & 0 & 0 \\ 0 & 0 & c & 0 \\ 0 & 0 & 0 & c \end{pmatrix}$ & $d_{RC}^2=4$ \\
\hline
4 & 1 & \centering $\left\langle \begin{pmatrix} 1 & 1 & 1 & 1 \end{pmatrix} \right\rangle$ & \centering $\begin{pmatrix} a & a & a & a \\ c & c & 0 & 0 \\ c & 0 & c & 0 \end{pmatrix}$ & $d_{RC}^4=0$ \\
\hline
4 & 2 & \centering $\left\langle \begin{pmatrix} 1 & 1 & 0 & 0 \\ 0 & 0 & 1 & 1 \end{pmatrix} \right\rangle$ & \centering $\begin{pmatrix} a & a & 0 & 0 \\ 0 & 0 & a & a \end{pmatrix}$ & $\begin{matrix} d_{RC}^2=0, \\ d_{RC}^4=0 \end{matrix}$ \\
\hline
5 & 1 & \centering $\left\langle \begin{pmatrix} 1 & 1 & 0 & 0 & 0 \end{pmatrix} \right\rangle$ & \centering $\begin{pmatrix} a & a & 0 & 0 & 0 \\ 0 & 0 & c & 0 & 0 \\ 0 & 0 & 0 & c & 0 \\ 0 & 0 & 0 & 0 & c \end{pmatrix}$ & $d_{RC}^2=4$ \\
\hline
5 & 1 & \centering $\left\langle \begin{pmatrix} 1 & 1 & 1 & 1 & 0 \end{pmatrix} \right\rangle$ & \centering $\begin{pmatrix} a & a & a & a & 0 \\ c & c & 0 & 0 & 0 \\ c & 0 & c & 0 & 0 \\ 0 & 0 & 0 & 0 & c \end{pmatrix}$ & $d_{RC}^4=2$ \\
\hline
5 & 2 & \centering $\left\langle \begin{pmatrix} 1 & 1 & 0 & 0 & 0 \\ 0 & 0 & 1 & 1 & 0 \end{pmatrix} \right\rangle$ & \centering $\begin{pmatrix} a & a & 0 & 0 & 0 \\ 0 & 0 & a & a & 0 \\ 0 & 0 & 0 & 0 & c \end{pmatrix}$ & $\begin{matrix} d_{RC}^2=2, \\ d_{RC}^4=2 \end{matrix}$ \\
\hline
6 & 1 & \centering $\left\langle \begin{pmatrix} 1 & 1 & 0 & 0 & 0 & 0 \end{pmatrix} \right\rangle$ & \centering $\begin{pmatrix} a & a & 0 & 0 & 0 & 0 \\ 0 & 0 & c & 0 & 0 & 0 \\ 0 & 0 & 0 & c & 0 & 0 \\ 0 & 0 & 0 & 0 & c & 0 \\ 0 & 0 & 0 & 0 & 0 & c \end{pmatrix}$ & $d_{RC}^2=4$ \\
\hline
\end{tabular}
}
\resizebox{7.2cm}{!}{
\begin{tabular}{ | c | c | c | c | c | }
\hline
$n$ & $k_1$ & \centering Residue code & \centering Generator Matrix & $d_{RC}^m$ \\
\hline
6 & 1 & \centering $\left\langle \begin{pmatrix} 1 & 1 & 1 & 1 & 0 & 0 \end{pmatrix} \right\rangle$ & \centering $\begin{pmatrix} a & a & a & a & 0 & 0 \\ c & c & 0 & 0 & 0 & 0 \\ c & 0 & c & 0 & 0 & 0 \\ 0 & 0 & 0 & 0 & c & 0 \\ 0 & 0 & 0 & 0 & 0 & c \end{pmatrix}$ & $d_{RC}^4=4$ \\
\hline
6 & 1 & \centering $\left\langle \begin{pmatrix} 1 & 1 & 1 & 1 & 1 & 1 \end{pmatrix} \right\rangle$ & \centering $\begin{pmatrix} a & a & a & a & a & a \\ c & c & 0 & 0 & 0 & 0 \\ c & 0 & c & 0 & 0 & 0 \\ c & 0 & 0 & c & 0 & 0 \\ c & 0 & 0 & 0 & c & 0 \end{pmatrix}$ & $d_{RC}^6=0$ \\
\hline
6 & 2 & \centering $\left\langle \begin{pmatrix} 1 & 1 & 0 & 0 & 0 & 0 \\ 0 & 0 & 1 & 1 & 0 & 0 \end{pmatrix} \right\rangle$ & \centering $\begin{pmatrix} a & a & 0 & 0 & 0 & 0 \\ 0 & 0 & a & a & 0 & 0 \\ 0 & 0 & 0 & 0 & c & 0 \\ 0 & 0 & 0 & 0 & 0 & c \end{pmatrix}$ & $\begin{matrix} d_{RC}^2=2, \\ d_{RC}^4=4 \end{matrix}$ \\
\hline
6 & 2 & \centering $\left\langle \begin{pmatrix} 1 & 1 & 1 & 1 & 0 & 0 \\ 0 & 0 & 0 & 0 & 1 & 1 \end{pmatrix} \right\rangle$ & \centering $\begin{pmatrix} a & a & a & a & 0 & 0 \\ 0 & 0 & 0 & 0 & a & a \\ c & c & 0 & 0 & 0 & 0 \\ c & 0 & c & 0 & 0 & 0 \end{pmatrix}$ & $\begin{matrix} d_{RC}^2=4, \\ d_{RC}^4=4, \\ d_{RC}^6=0 \end{matrix}$ \\
\hline
6 & 2 & \centering $\left\langle \begin{pmatrix} 1 & 1 & 1 & 1 & 0 & 0 \\ 0 & 0 & 1 & 1 & 1 & 1 \end{pmatrix} \right\rangle$ & \centering $\begin{pmatrix} a & a & a & a & 0 & 0 \\ 0 & 0 & a & a & a & a \\ 0 & 0 & c & c & 0 & 0 \\ c & 0 & c & 0 & c & 0 \end{pmatrix}$ & $d_{RC}^4=4$ \\
\hline
6 & 2 & \centering $\left\langle \begin{pmatrix} 1 & 1 & 0 & 0 & 0 & 0 \\ 0 & 0 & 1 & 1 & 0 & 0 \\ 0 & 0 & 0 & 0 & 1 & 1 \end{pmatrix} \right\rangle$ & \centering $\begin{pmatrix} a & a & 0 & 0 & 0 & 0 \\ 0 & 0 & a & a & 0 & 0 \\ 0 & 0 & 0 & 0 & a & a \end{pmatrix}$ & $\begin{matrix} d_{RC}^2=2, \\ d_{RC}^4=2, \\ d_{RC}^6=0 \end{matrix}$ \\
\hline
7 & 1 & \centering $\left\langle \begin{pmatrix} 1 & 1 & 0 & 0 & 0 & 0 & 0 \end{pmatrix} \right\rangle$ & \centering $\begin{pmatrix} a & a & 0 & 0 & 0 & 0 & 0 \\ 0 & 0 & c & 0 & 0 & 0 & 0 \\ 0 & 0 & 0 & c & 0 & 0 & 0 \\ 0 & 0 & 0 & 0 & c & 0 & 0 \\ 0 & 0 & 0 & 0 & 0 & c & 0 \\ 0 & 0 & 0 & 0 & 0 & 0 & c \end{pmatrix}$ & $d_{RC}^2=4$ \\
\hline
\end{tabular}
}
\resizebox{7.2cm}{!}{
\begin{tabular}{ | c | c | c | c | c | }
\hline
$n$ & $k_1$ & \centering Residue code & \centering Generator Matrix & $d_{RC}^m$ \\
\hline
7 & 1 & \centering $\left\langle \begin{pmatrix} 1 & 1 & 1 & 1 & 0 & 0 & 0 \end{pmatrix} \right\rangle$ & \centering $\begin{pmatrix} a & a & a & a & 0 & 0 & 0 \\ c & c & 0 & 0 & 0 & 0 & 0 \\ c & 0 & c & 0 & 0 & 0 & 0 \\ 0 & 0 & 0 & 0 & c & 0 & 0 \\ 0 & 0 & 0 & 0 & 0 & c & 0 \\ 0 & 0 & 0 & 0 & 0 & 0 & c \end{pmatrix}$ & $d_{RC}^4=6$ \\
\hline
7 & 1 & \centering $\left\langle \begin{pmatrix} 1 & 1 & 1 & 1 & 1 & 1 & 0 \end{pmatrix} \right\rangle$ & \centering $\begin{pmatrix} a & a & a & a & a & a & 0 \\ c & c & 0 & 0 & 0 & 0 & 0 \\ c & 0 & c & 0 & 0 & 0 & 0 \\ c & 0 & 0 & c & 0 & 0 & 0 \\ c & 0 & 0 & 0 & c & 0 & 0 \\ 0 & 0 & 0 & 0 & 0 & 0 & c \end{pmatrix}$ & $d_{RC}^6=2$ \\
\hline
7 & 2 & \centering $\left\langle \begin{pmatrix} 1 & 1 & 0 & 0 & 0 & 0 & 0 \\ 0 & 0 & 1 & 1 & 0 & 0 & 0 \end{pmatrix} \right\rangle$ & \centering $\begin{pmatrix} a & a & 0 & 0 & 0 & 0 & 0 \\ 0 & 0 & a & a & 0 & 0 & 0 \\ 0 & 0 & 0 & 0 & c & 0 & 0 \\ 0 & 0 & 0 & 0 & 0 & c & 0 \\ 0 & 0 & 0 & 0 & 0 & 0 & c \end{pmatrix}$ & $\begin{matrix} d_{RC}^2=2, \\ d_{RC}^4=6 \end{matrix}$ \\
\hline
7 & 2 & \centering $\left\langle \begin{pmatrix} 1 & 1 & 1 & 1 & 0 & 0 & 0 \\ 0 & 0 & 0 & 0 & 1 & 1 & 0 \end{pmatrix} \right\rangle$ & \centering $\begin{pmatrix} a & a & a & a & 0 & 0 & 0 \\ 0 & 0 & 0 & 0 & a & a & 0 \\ c & c & 0 & 0 & 0 & 0 & 0 \\ c & 0 & c & 0 & 0 & 0 & 0 \\ 0 & 0 & 0 & 0 & 0 & 0 & c \end{pmatrix}$ & $\begin{matrix} d_{RC}^2=4, \\ d_{RC}^4=6, \\ d_{RC}^6=2 \end{matrix}$ \\
\hline
7 & 2 & \centering $\left\langle \begin{pmatrix} 1 & 1 & 1 & 1 & 0 & 0 & 0 \\ 0 & 0 & 1 & 1 & 1 & 1 & 0 \end{pmatrix} \right\rangle$ & \centering $\begin{pmatrix} a & a & a & a & 0 & 0 & 0 \\ 0 & 0 & a & a & a & a & 0 \\ 0 & 0 & c & c & 0 & 0 & 0 \\ c & 0 & c & 0 & c & 0 & 0 \\ 0 & 0 & 0 & 0 & 0 & 0 & c \end{pmatrix}$ & $d_{RC}^4=4$ \\
\hline
7 & 3 & \centering $\left\langle \begin{pmatrix} 1 & 1 & 0 & 0 & 0 & 0 & 0 \\ 0 & 0 & 1 & 1 & 0 & 0 & 0 \\ 0 & 0 & 0 & 0 & 1 & 1 & 0 \end{pmatrix} \right\rangle$ & \centering $\begin{pmatrix} a & a & 0 & 0 & 0 & 0 & 0 \\ 0 & 0 & a & a & 0 & 0 & 0 \\ 0 & 0 & 0 & 0 & a & a & 0 \\ 0 & 0 & 0 & 0 & 0 & 0 & c \end{pmatrix}$ & $\begin{matrix} d_{RC}^2=4, \\ d_{RC}^4=2, \\ d_{RC}^6=2 \end{matrix}$ \\
\hline
\end{tabular}
}
\end{adjustbox}
\end{table}

\begin{table}
\begin{adjustbox}{angle=90}
\resizebox{7.2cm}{!}{
\begin{tabular}{ | c | c | c | c | c | }
\hline
$n$ & $k_1$ & \centering Residue code & \centering Generator Matrix & $d_{RC}^m$ \\
\hline
8 & 1 & \centering $\left\langle \begin{pmatrix} 1 & 1 & 0 & 0 & 0 & 0 & 0 & 0 \end{pmatrix} \right\rangle$ & \centering $\begin{pmatrix} a & a & 0 & 0 & 0 & 0 & 0 & 0 \\ 0 & 0 & c & 0 & 0 & 0 & 0 & 0 \\ 0 & 0 & 0 & c & 0 & 0 & 0 & 0 \\ 0 & 0 & 0 & 0 & c & 0 & 0 & 0 \\ 0 & 0 & 0 & 0 & 0 & c & 0 & 0 \\ 0 & 0 & 0 & 0 & 0 & 0 & c & 0 \\ 0 & 0 & 0 & 0 & 0 & 0 & 0 & c \end{pmatrix}$ & $d_{RC}^2=4$ \\
\hline
8 & 1 & \centering $\left\langle \begin{pmatrix} 1 & 1 & 1 & 1 & 0 & 0 & 0 & 0 \end{pmatrix} \right\rangle$ & \centering $\begin{pmatrix} a & a & a & a & 0 & 0 & 0 & 0 \\ c & c & 0 & 0 & 0 & 0 & 0 & 0 \\ c & 0 & c & 0 & 0 & 0 & 0 & 0 \\ 0 & 0 & 0 & 0 & c & 0 & 0 & 0 \\ 0 & 0 & 0 & 0 & 0 & c & 0 & 0 \\ 0 & 0 & 0 & 0 & 0 & 0 & c & 0 \\ 0 & 0 & 0 & 0 & 0 & 0 & 0 & c \end{pmatrix}$ & $d_{RC}^4=8$ \\
\hline
8 & 1 & \centering $\left\langle \begin{pmatrix} 1 & 1 & 1 & 1 & 1 & 1 & 0 & 0 \end{pmatrix} \right\rangle$ & \centering $\begin{pmatrix} a & a & a & a & a & a & 0 & 0 \\ c & c & 0 & 0 & 0 & 0 & 0 & 0 \\ c & 0 & c & 0 & 0 & 0 & 0 & 0 \\ c & 0 & 0 & c & 0 & 0 & 0 & 0 \\ c & 0 & 0 & 0 & c & 0 & 0 & 0 \\ 0 & 0 & 0 & 0 & 0 & 0 & c & 0 \\ 0 & 0 & 0 & 0 & 0 & 0 & 0 & c \end{pmatrix}$ & $d_{RC}^6=4$ \\
\hline
8 & 1 & \centering $\left\langle \begin{pmatrix} 1 & 1 & 1 & 1 & 1 & 1 & 1 & 1 \end{pmatrix} \right\rangle$ & \centering $\begin{pmatrix} a & a & a & a & a & a & a & a \\ c & c & 0 & 0 & 0 & 0 & 0 & 0 \\ c & 0 & c & 0 & 0 & 0 & 0 & 0 \\ c & 0 & 0 & c & 0 & 0 & 0 & 0 \\ c & 0 & 0 & 0 & c & 0 & 0 & 0 \\ c & 0 & 0 & 0 & 0 & c & 0 & 0 \\ c & 0 & 0 & 0 & 0 & 0 & c & 0 \end{pmatrix}$ & $d_{RC}^8=0$ \\
\hline
8 & 2 & \centering $\left\langle \begin{pmatrix} 1 & 1 & 0 & 0 & 0 & 0 & 0 & 0 \\ 0 & 0 & 1 & 1 & 0 & 0 & 0 & 0 \end{pmatrix} \right\rangle$ & \centering $\begin{pmatrix} a & a & 0 & 0 & 0 & 0 & 0 & 0 \\ 0 & 0 & a & a & 0 & 0 & 0 & 0 \\ 0 & 0 & 0 & 0 & c & 0 & 0 & 0 \\ 0 & 0 & 0 & 0 & 0 & c & 0 & 0 \\ 0 & 0 & 0 & 0 & 0 & 0 & c & 0 \\ 0 & 0 & 0 & 0 & 0 & 0 & 0 & c \end{pmatrix}$ & $\begin{matrix} d_{RC}^2=4, \\ d_{RC}^4=8\end{matrix}$ \\
\hline
\end{tabular}
}
\resizebox{7.2cm}{!}{
\begin{tabular}{ | c | c | c | c | c | }
\hline
$n$ & $k_1$ & \centering Residue code & \centering Generator Matrix & $d_{RC}^m$ \\
\hline
8 & 2 & \centering $\left\langle \begin{pmatrix} 1 & 1 & 1 & 1 & 0 & 0 & 0 & 0 \\ 0 & 0 & 0 & 0 & 1 & 1 & 0 & 0 \end{pmatrix} \right\rangle$ & \centering $\begin{pmatrix} a & a & a & a & 0 & 0 & 0 & 0 \\ 0 & 0 & 0 & 0 & a & a & 0 & 0 \\ c & c & 0 & 0 & 0 & 0 & 0 & 0 \\ c & 0 & c & 0 & 0 & 0 & 0 & 0 \\ 0 & 0 & 0 & 0 & 0 & 0 & c & 0 \\ 0 & 0 & 0 & 0 & 0 & 0 & 0 & c \end{pmatrix}$ & $\begin{matrix} d_{RC}^2=4, \\ d_{RC}^4=8, \\ d_{RC}^6=4 \end{matrix}$ \\
\hline
8 & 2 & \centering $\left\langle \begin{pmatrix} 1 & 1 & 1 & 1 & 1 & 1 & 0 & 0 \\ 0 & 0 & 0 & 0 & 0 & 0 & 1 & 1 \end{pmatrix} \right\rangle$ & \centering $\begin{pmatrix} a & a & a & a & a & a & 0 & 0 \\ 0 & 0 & 0 & 0 & 0 & 0 & a & a \\ c & c & 0 & 0 & 0 & 0 & 0 & 0 \\ c & 0 & c & 0 & 0 & 0 & 0 & 0 \\ c & 0 & 0 & c & 0 & 0 & 0 & 0 \\ c & 0 & 0 & 0 & c & 0 & 0 & 0 \\ \end{pmatrix}$ & $\begin{matrix} d_{RC}^2=4, \\ d_{RC}^6=4, \\ d_{RC}^8=0 \end{matrix}$ \\
\hline
8 & 2 & \centering $\left\langle \begin{pmatrix} 1 & 1 & 1 & 1 & 0 & 0 & 0 & 0 \\ 0 & 0 & 0 & 0 & 1 & 1 & 1 & 1 \end{pmatrix} \right\rangle$ & \centering $\begin{pmatrix} a & a & a & a & 0 & 0 & 0 & 0 \\ 0 & 0 & 0 & 0 & a & a & a & a \\ c & c & 0 & 0 & 0 & 0 & 0 & 0 \\ c & 0 & c & 0 & 0 & 0 & 0 & 0 \\ 0 & 0 & 0 & 0 & c & c & 0 & 0 \\ 0 & 0 & 0 & 0 & c & 0 & c & 0 \\ \end{pmatrix}$ & $\begin{matrix} d_{RC}^4=0, \\ d_{RC}^8=0 \end{matrix}$ \\
\hline
8 & 2 & \centering $\left\langle \begin{pmatrix} 1 & 1 & 1 & 1 & 0 & 0 & 0 & 0 \\ 0 & 0 & 1 & 1 & 1 & 1 & 0 & 0 \end{pmatrix} \right\rangle$ & \centering $\begin{pmatrix} a & a & a & a & 0 & 0 & 0 & 0 \\ 0 & 0 & a & a & a & a & 0 & 0 \\ 0 & 0 & c & c & 0 & 0 & 0 & 0 \\ c & 0 & c & 0 & c & 0 & 0 & 0 \\ 0 & 0 & 0 & 0 & 0 & 0 & c & 0 \\ 0 & 0 & 0 & 0 & 0 & 0 & 0 & c \end{pmatrix}$ & $d_{RC}^4=4$ \\
\hline
8 & 2 & \centering $\left\langle \begin{pmatrix} 1 & 1 & 1 & 1 & 1 & 1 & 0 & 0 \\ 0 & 0 & 0 & 0 & 1 & 1 & 1 & 1 \end{pmatrix} \right\rangle$ & \centering $\begin{pmatrix} a & a & a & a & a & a & 0 & 0 \\ 0 & 0 & 0 & 0 & a & a & a & a \\ c & c & 0 & 0 & 0 & 0 & 0 & 0 \\ c & 0 & c & 0 & 0 & 0 & 0 & 0 \\ 0 & 0 & 0 & 0 & c & c & 0 & 0 \\ c & 0 & 0 & 0 & c & 0 & c & 0 \end{pmatrix}$ & $\begin{matrix} d_{RC}^4=8, \\ d_{RC}^6=4 \end{matrix}$ \\
\hline
8 & 3 & \centering $\left\langle \begin{pmatrix} 1 & 1 & 0 & 0 & 0 & 0 & 0 & 0 \\ 0 & 0 & 1 & 1 & 0 & 0 & 0 & 0 \\ 0 & 0 & 0 & 0 & 1 & 1 & 0 & 0 \end{pmatrix} \right\rangle$ & \centering $\begin{pmatrix} a & a & 0 & 0 & 0 & 0 & 0 & 0 \\ 0 & 0 & a & a & 0 & 0 & 0 & 0 \\ 0 & 0 & 0 & 0 & a & a & 0 & 0 \\ 0 & 0 & 0 & 0 & 0 & 0 & c & 0 \\ 0 & 0 & 0 & 0 & 0 & 0 & 0 & c \end{pmatrix}$ & $\begin{matrix} d_{RC}^2=2, \\ d_{RC}^4=4, \\ d_{RC}^6=2 \end{matrix}$ \\
\hline
\end{tabular}
}
\resizebox{7.2cm}{!}{
\begin{tabular}{ | c | c | c | c | c | }
\hline
$n$ & $k_1$ & \centering Residue code & \centering Generator Matrix & $d_{RC}^m$ \\
\hline

8 & 3 & \centering $\left\langle \begin{pmatrix} 1 & 1 & 1 & 1 & 0 & 0 & 0 & 0 \\ 0 & 0 & 0 & 0 & 1 & 1 & 0 & 0 \\ 0 & 0 & 0 & 0 & 0 & 0 & 1 & 1 \end{pmatrix} \right\rangle$ & \centering $\begin{pmatrix} a & a & a & a & 0 & 0 & 0 & 0 \\ 0 & 0 & 0 & 0 & a & a & 0 & 0 \\ 0 & 0 & 0 & 0 & 0 & 0 & a & a \\ c & c & 0 & 0 & 0 & 0 & 0 & 0 \\ c & 0 & c & 0 & 0 & 0 & 0 & 0 \\ \end{pmatrix}$ & $\begin{matrix} d_{RC}^2=4, \\ d_{RC}^4=0, \\ d_{RC}^6=4, \\ d_{RC}^8=0 \end{matrix}$ \\
\hline
8 & 3 & \centering $\left\langle \begin{pmatrix} 1 & 1 & 1 & 1 & 0 & 0 & 0 & 0 \\ 0 & 0 & 1 & 1 & 1 & 1 & 0 & 0 \\ 0 & 0 & 0 & 0 & 0 & 0 & 1 & 1 \end{pmatrix} \right\rangle$ & \centering $\begin{pmatrix} a & a & a & a & 0 & 0 & 0 & 0 \\ 0 & 0 & a & a & a & a & 0 & 0 \\ 0 & 0 & 0 & 0 & 0 & 0 & a & a \\ 0 & 0 & c & c & 0 & 0 & 0 & 0 \\ c & 0 & c & 0 & c & 0 & 0 & 0 \end{pmatrix}$ & $\begin{matrix} d_{RC}^2=4, \\ d_{RC}^4=4, \\ d_{RC}^6=4 \end{matrix}$ \\
\hline
8 & 3 & \centering $\left\langle \begin{pmatrix} 1 & 0 & 0 & 0 & 1 & 1 & 1 & 0 \\ 0 & 1 & 0 & 1 & 0 & 1 & 1 & 0 \\ 0 & 0 & 1 & 1 & 1 & 0 & 1 & 0 \end{pmatrix} \right\rangle$ & \centering $\begin{pmatrix} a & 0 & 0 & 0 & a & a & a & 0 \\ 0 & a & 0 & a & 0 & a & a & 0 \\ 0 & 0 & a & a & a & 0 & a & 0 \\ c & c & 0 & 0 & 0 & c & 0 & 0 \\ c & 0 & c & 0 & c & 0 & 0 & 0 \\ 0 & c & c & c & 0 & 0 & 0 & 0 \\ 0 & 0 & 0 & 0 & 0 & 0 & 0 & c \end{pmatrix}$ & $d_{RC}^4=4$ \\
\hline
8 & 3 & \centering $\left\langle \begin{pmatrix} 1 & 1 & 0 & 0 & 0 & 0 & 1 & 1 \\ 0 & 0 & 1 & 1 & 0 & 0 & 1 & 1 \\ 0 & 0 & 0 & 0 & 1 & 1 & 1 & 1 \end{pmatrix} \right\rangle$ & \centering $\begin{pmatrix} a & a & 0 & 0 & 0 & 0 & a & a \\ 0 & 0 & a & a & 0 & 0 & a & a \\ 0 & 0 & 0 & 0 & a & a & a & a \\ 0 & 0 & 0 & 0 & 0 & 0 & c & c \\ c & 0 & c & 0 & c & 0 & c & 0 \end{pmatrix}$ & $\begin{matrix} d_{RC}^4=0, \\ d_{RC}^6=4 \end{matrix}$ \\
\hline
8 & 4 & \centering $\left\langle \begin{pmatrix} 1 & 1 & 0 & 0 & 0 & 0 & 0 & 0 \\ 0 & 0 & 1 & 1 & 0 & 0 & 0 & 0 \\ 0 & 0 & 0 & 0 & 1 & 1 & 0 & 0 \\ 0 & 0 & 0 & 0 & 0 & 0 & 1 & 1 \end{pmatrix} \right\rangle$ & \centering $\begin{pmatrix} a & a & 0 & 0 & 0 & 0 & 0 & 0 \\ 0 & 0 & a & a & 0 & 0 & 0 & 0 \\ 0 & 0 & 0 & 0 & a & a & 0 & 0 \\ 0 & 0 & 0 & 0 & 0 & 0 & a & a \end{pmatrix}$ & $\begin{matrix} d_{RC}^2=2, \\ d_{RC}^4=0, \\ d_{RC}^6=2, \\ d_{RC}^8=0 \end{matrix}$ \\
\hline
8 & 4 & \centering $\left\langle \begin{pmatrix} 1 & 0 & 0 & 0 & 1 & 1 & 1 & 0 \\ 0 & 1 & 0 & 0 & 1 & 1 & 0 & 1 \\ 0 & 0 & 1 & 0 & 1 & 0 & 1 & 1 \\ 0 & 0 & 0 & 1 & 0 & 1 & 1 & 1 \end{pmatrix} \right\rangle$ & \centering $\begin{pmatrix} a & 0 & 0 & 0 & a & a & a & 0 \\ 0 & a & 0 & 0 & a & a & 0 & a \\ 0 & 0 & a & 0 & a & 0 & a & a \\ 0 & 0 & 0 & a & 0 & a & a & a \\ c & c & c & 0 & c & 0 & 0 & 0 \\ c & c & 0 & c & 0 & c & 0 & 0 \\ c & 0 & c & c & 0 & 0 & c & 0 \\ 0 & c & c & c & 0 & 0 & 0 & c \end{pmatrix}$ & $d_{RC}^4=0$ \\
\hline
\end{tabular}
}
\end{adjustbox}
\end{table}

\nocite{King, Gaborit}


\begin{thebibliography}{00}


\bibitem{Adleman} Adleman L. M.: Molecular computation of solutions to combinatorial problems. Science. {\bf 266}(5187), 1021--1024(1994)

\bibitem{Alahmadi2} Alahmadi A., Alkathiry A., Altassan A., Basaffar W., Bonnecaze A., Shoaib H., Sol{\'e} P.: Type IV codes over a non-local non-unital ring. Proyecciones (Antofagasta). {\bf 39}(4), 963--978(2020)

\bibitem{Alahmadi1} Alahmadi A., Altassan A., Basaffar W., Bonnecaze A., Shoaib H., Sol{\'e} P.: Type VI codes over a non-unital ring. to appear in J. Algebra Its Appl. Available from
{\tt https://hal.archives-ouvertes.fr/hal-02433480/document}.

\bibitem{Alahmadi3} Alahmadi A., Altassan A.,  Basaffar W., Bonnecaze A., Shoaib H., Sol{\'e} P.: Quasi Type IV codes over a non-unital ring. preprint available from

    {\tt https://hal.archives-ouvertes.fr/hal-02544399/document}.


\bibitem{Bennenni} Bennenni N., Guenda K., Mesnager S.: New DNA cyclic codes over rings. Adv. Math. Commum. {\bf 11}(1) 83--98(2017)


\bibitem{Bouyukliev} Bouyukliev I., Bouyuklieva S.,  Gulliver T. A., Ostergard P. R. J.: Classification of optimal binary self-orthogonal codes. J. Comb. Math. Comb. Comput. {\bf 59} 33-87(2006)

\bibitem{Fine} Fine B.: Classification of finite rings of order $p^2$. Math. Mag. {\bf 66}(4) 248--252(1993)

\bibitem{Gaborit} Gaborit P., King O. D.: Linear constructions for DNA codes. Theor. Comput. Sci. {\bf 334}(1-3) 99--113(2005)

\bibitem{Guenda} Guenda K., Gulliver T. A.: Construction of cyclic codes over {$\mathbb{F}_2 + u \mathbb{F}_2$} for DNA computing. Appl. Algebra Eng. Commun. {\bf 24}(6) 445--459(2013)

\bibitem{Hou} Hou X.-D.: On the number of inequivalent binary self-orthogonal codes. IEEE Trans. Inform. Theory. {\bf 53}(7) 2459--2479(2007)

\bibitem{KimWeb} J.-L. Kim's CICAGO Lab website

  {\tt https://cicagolab.sogang.ac.kr/cicagolab/2656.html}.

\bibitem{King} King O. D.: Bounds for DNA codes with constant GC-content. Electron. J. Comb. {\bf 10} R33(2003)


\bibitem{Liang} Liang J., Wang L.: On cyclic DNA codes over $\mathbb{F}_2 + u \mathbb{F}_2$. J. Comput. Appl. Math. {\bf 51}(1-2)  81--91(2016)


\bibitem{Limbachiya} Limbachiya D., Rao B., Gupta M. K.: The art of DNA strings: sixteen years of DNA coding theory. {\tt https://arxiv.org/pdf/1607.00266.pdf}


\bibitem{MacWilliams} MacWilliams F. S., Sloane N. J. A.: The Theory of Error-Correcting Codes. Elsevier. {\bf 16} (1977)

\bibitem{Mil} Milenkovic O., Kashyap N.:
On the design of codes for DNA computing. Intern. Workshop
on Coding and Cryptography, 100-119(2005)

\bibitem{Pless} Pless V.: A classification of self-orthogonal codes over $GF(2)$. Discrete Math. {\bf 3}(1-3) 209--246(1972)


\bibitem{Siap} Siap I., Abualrub T., Ghrayeb A.:
        Cyclic DNA codes over the ring {$\mathbb{F}_2 [u]/(u^2-1)$} based on the deletion distance. J. Franklin Inst. {\bf 346}(8) 731--740(2009)


\end{thebibliography}
\end{document}